\documentclass[11pt]{article}

\usepackage{authblk}
\usepackage{amsfonts}
\usepackage{fancyhdr}
\usepackage{comment}
\usepackage{hyperref}
\usepackage{mdframed}
\usepackage[a4paper, top=2.5cm, bottom=2.5cm, left=2.2cm, right=2.2cm]{geometry}
\usepackage{amsmath}
\usepackage{amsthm}
\usepackage{changepage}
\usepackage{natbib}
\usepackage{graphicx}%
\usepackage{cleveref}
\usepackage{tablefootnote}
\usepackage{multirow}
\usepackage{pgfplots}
\usepackage{caption}
\usepackage{subcaption}
\usepackage{mathtools}

\pgfplotsset{width=10cm,compat=1.9}

\setcitestyle{parens}

\newtheorem{theorem}{Theorem}

\newtheorem{algorithm}{Algorithm}

\newtheorem{corollary}{Corollary}

\newtheorem{definition}{Definition}

\newtheorem{lemma}{Lemma}

\newtheorem{remark}{Remark}

\usepackage{graphicx}
\graphicspath{ {./images/} }

\newcommand{\cumval}{V}
\newcommand{\val}{v}
\newcommand{\util}{u}
\newcommand{\cumutil}{U}
\newcommand{\icumutil}{\overline{\cumutil}}
\newcommand{\icumval}{\overline{\cumval}}
\newcommand{\iutil}{\overline{\util}}

\DeclareMathOperator{\OPT}{OPT}
\DeclareMathOperator{\DA}{DA}
\DeclareMathOperator{\EAR}{EAR}
\newcommand{\OPTCS}{\OPT^{\cumutil}}
\newcommand{\OPTOM}{\OPT^{\cumval}}
\newcommand{\OPTSEQ}[1]{\OPT^{SEQ}}
\newcommand{\DACS}{\DA^{\cumutil}}
\newcommand{\DAOM}{\DA^{\cumval}}
\newcommand{\bridgemechanism}[1]{\hat{y}}
\newcommand{\bridgeallocation}[1]{\hat{y}_{#1}}
\newcommand{\simplex}{\Delta}
\DeclareMathOperator{\GSP}{SP}
\newcommand{\epsGSP}[1]{\GSP(#1, \epsilon)}

\newcommand{\apxda}{\gamma}
\newcommand{\apxcs}{\zeta}
\newcommand{\apxom}{\pi}
\newcommand{\apxrough}{\alpha}
\newcommand{\feasibility}{\mathcal{P}}
\newcommand{\dcfeasibility}{\overline{\feasibility}}

\begin{document}

\title{Combinatorial Pen Testing (or Consumer Surplus of Deferred-Acceptance Auctions)}
\author[1]{Aadityan Ganesh}
\author[2]{Jason Hartline}
\affil[1]{Princeton University, aadityanganesh@princeton.edu}
\affil[2]{Northwestern University, hartline@northwestern.edu}

\date{}

\maketitle

\begin{abstract}
Pen testing is the problem of selecting high-capacity resources when
the only way to measure the capacity of a resource expends its
capacity. We have a set of $n$ pens with unknown amounts of ink and our goal is to select a feasible subset of pens maximizing the total ink in them. We are allowed to learn about the ink levels by writing with them, but this uses up ink that was previously in the pens. 

We identify optimal and near optimal pen testing algorithms by drawing analogues to auction theoretic frameworks of deferred-acceptance auctions and virtual values. Our framework allows the conversion of any near optimal deferred-acceptance mechanism into a near optimal pen testing algorithm.
Moreover, these algorithms guarantee an additional overhead of at most $(1+o(1)) \ln n$ in the approximation factor to the omniscient algorithm that has access to the ink levels in the pens.
We use this framework to give pen testing algorithms for various combinatorial constraints like matroid, knapsack, and general downward-closed constraints, and also for online environments.
\end{abstract}

\section{Introduction}

Pen testing is the problem of selecting high capacity resources when
the only way to measure the capacity of the resource expends its
capacity \citep{QV23}.  We show that any ascending-price auction can be converted into an equivalent pen testing algorithm. We apply the auction theoretic
framework of virtual values to identify optimal pen testing
algorithms.  This connection allows many existing results from
auction theory to be applied to pen testing and gives optimal and near
optimal pen testing algorithms in combinatorial and online environments.

The pen testing problem of \citet{QV23} is the following. We have a set of pens with varying amounts of remaining ink and we want to choose a pen with the largest amount of ink left.
We have access to the distribution of ink levels in these pens, but we are only allowed to learn about the ink levels by writing with the pens to test them.
While writing gives information about whether there is still ink left, it uses up ink that was previously in the pen. Pens that are expended due to testing can be discarded without incurring any penalty.
%We compare the performance of our algorithms against (a) the optimal pen testing algorithm and (b) the optimal algorithm that knows the amount of ink in each pen, which we call the \emph{omniscient benchmark}.
The combinatorial pen testing problem generalizes the pen testing problem to selecting a feasible subset of pens, e.g., according to a matroid or knapsack constraint, to maximize the total remaining ink in the chosen subset. 
%\Anote{To motivate the combinatorial generalization, consider the knapsack pen testing problem. Pens come in different shapes and sizes. We want to maximize the total residual ink in the chosen pens while being able to fit them inside our pen pouch.}.

Deferred-acceptance auctions \citep{MS14} formalize the notion of ascending-price mechanisms for allocating goods to strategic agents.
These are mechanisms that greedily reject the least promising agent by increasing the price for getting allocated at each stage. For instance, to auction off one good, the deferred-acceptance mechanism keeps increasing the price, rejecting agents whenever their value for the good falls below the price until exactly one agent remains in the auction. The \emph{surplus} of the mechanism is the sum of the values of all the allocated agents and the \emph{consumer surplus} is the surplus minus the payments made by the agents.

We provide a black box reduction to convert any surplus optimal (or near optimal) deferred-acceptance auction into an optimal (or near optimal) pen testing algorithm with equivalent performance guarantees.
In this reduction, pens and their ink levels are analogous to agents and their values.
Payments correspond to the ink used up via testing, and as a consequence, the consumer surplus of the deferred-acceptance auction corresponds to the residual ink in the pen testing environment.
Therefore, optimizing residual ink reduces to designing consumer-surplus-optimal deferred-acceptance mechanisms.

\citet{QV23} propose comparing the performance of a given pen testing algorithm against the \emph{omniscient benchmark}, the performance of an omniscient algorithm that has access to the ink levels in the pens.
The omniscient benchmark can thereby choose the optimal feasible subset of pens without having to burn any ink (i.e, payments in the auction environment), and is thus, analogous to the surplus optimal (sealed-bid) mechanism in the corresponding auction environment.
Thus, the \emph{omniscient approximation}, the ratio between the omniscient benchmark to the residual ink in the pen testing algorithm, is analogous to the ratio between the optimal surplus and the consumer surplus of the corresponding deferred-acceptance auction.

%Our reduction from pen testing to deferred-acceptance auction design allow us to draw a natural correspondence between optimal mechanisms and optimal pen testing and omniscient algorithms.
%Comparing against the optimal pen testing algorithm is analogous to upper bounding the ratio between the consumer surplus optimal deferred-acceptance mechanism and the consumer surplus of the underlying deferred-acceptance auction.

In the literature, near optimal deferred-acceptance mechanisms are evaluated against optimal (sealed bid) mechanisms, an upper bound to the optimal deferred-acceptance mechanism.
Define the \emph{standard benchmark} as the optimal (sealed-bid) mechanism for consumer surplus and let the \emph{standard approximation} $\apxda(n)$ be the ratio between the standard benchmark and the consumer surplus of a given deferred-acceptance mechanism.
Note that the standard benchmark does not have a natural pen testing analog, but is a much tighter upper bound on the optimal pen testing algorithm than the omniscient benchmark.
Thus, the reduction to deferred-acceptance mechanism design enables a much tighter approximation bound for the optimal pen testing algorithm.

The reduction framework decomposes into two steps for designing near optimal deferred-acceptance mechanisms.
One of the two steps involves proving upper bounds on the standard approximation $\apxda(n)$.
In the other step, for any given feasibility constraint, we consider the \emph{optimal omniscient approximation}, the ratio between the omniscient benchmark (here, the optimal surplus) and the standard benchmark (here, the optimal consumer surplus).
As we will argue in \Cref{thm:IntroSurplusDA}, the omniscient approximation $\apxom(n)$ is at most the product $\apxda(n) \, \apxcs(n)$ between the two approximation ratios.
At a high level, $\apxda(n)$ captures the loss in surplus due to running a deferred-acceptance auction and $\apxcs(n)$ captures the loss in surplus due to burning all payments.
Beyond consumer surplus optimization, the reduction framework applies more generally to optimization problems with two natural benchmarks, a weaker standard benchmark and much more powerful omniscient benchmark.

\begin{theorem} \label{thm:IntroSurplusDA}
Consider a combinatorial pen testing environment with $n$ pens, the analogous auction environment with an optimal omniscient approximation $\apxcs(n)$, and a deferred-acceptance mechanism with a standard approximation equal to $\apxda(n)$.
Then, there is a pen testing algorithm that is a $\apxda(n)$-approximation to the standard benchmark and $\apxom(n) = \apxda(n) \, \apxcs(n)$-approximation to the omniscient benchmark.\end{theorem}

\begin{comment}
\begin{theorem} \label{thm:IntroSurplusDA}
    Consider a combinatorial pen testing environment with $n$ pens, the analogous auction environment, and a deferred acceptance mechanism $\DA$ for the auction environment.  Let $\apxda(n)$ denote the {\em standard approximation} of $\DA$, i.e., the worst-case ratio of the expected optimal surplus to the expected surplus of $\DA$ for the auction environment.  Let $\apxcs(n)$ denote the {\em optimal omniscient approximation}, i.e., the worst-case ratio of the omniscient benchmark to the expected optimal consumer surplus for the auction environment. Then, there is a pen testing algorithm that is a $\apxda(n)$-approximation to the standard benchmark and $\apxom(n) = \apxda(n) \, \apxcs(n)$-approximation to the omniscient benchmark.
\end{theorem}    

\Cref{thm:IntroSurplusDA} gives a reduction framework for designing near optimal pen testing algorithms under arbitrary feasibility constraints. This approach improves bounds from previous literature for all the environments that we discuss; however, it does not always yield a tight bound. For the online IID environment (described below), we identify an algorithm with omniscient approximation $\apxom(n)$ that is a constant factor less than $\apxda(n) \, \apxcs(n)$.

\end{comment}

In the first step of our framework, we bound the optimal omniscient approximation $\apxcs(n)$ under various feasibility constraints.
Generally, remember that the optimal omniscient approximation $\apxcs(n)$ is the ratio between the omniscient benchmark and the standard benchmark, and for our specific application, it corresponds to the loss in surplus due to payments getting burnt instead of getting transferred to the auctioneer.
We show that $\apxcs(n) \leq (1 + o(1)) \ln n$ in any combinatorial environment.
Our result tightens the best known upper bound on $\apxcs(n)$ for most settings of interest.
\citet{FTT16} show that the optimal omniscient approximation is at most $\tfrac{2}{\ln 2} \, (1 + o(1)) \ln m$, where $m$ is the number of feasible outcomes.
Our bound improves their bound by removing the $\tfrac{2}{\ln 2}$ constant and by replacing $m$ by $n$, which is usually smaller.\footnote{In some non-downward-closed settings like public projects, $m < n$. For example, $m=2$ for building a bridge that serves everyone. Either the bridge is built or not. However, $m \geq n$ in all downward-closed environments.
Indeed, for each bidder, we assume there exists an outcome where they are included; otherwise, the bidder can be dropped from consideration altogether.} \citet{HR08} show that $\apxcs(n) \leq \frac{2}{\ln 2} (1+o(1)) \ln \frac{n}{k}$ when any $k$ out of the $n$ agents can be allocated the good simultaneously, i.e, any $k$ pens can be chosen in the corresponding pen testing environment.
We improve their bound by a multiplicative factor of at least $1.36 (1 - o(1))$. We also prove that $\apxcs(n) \geq (0.577 + \ln \frac{n}{k})$ for the $k$-identical goods setting, which we show is tight up to $(1 + o(1))$ terms when $k = o(n)$.

Our upper bound for $\apxcs(n)$ is via a reduction from multi-agent environments to a single-agent environment, similar in flavour to \citet{FHL20}.
In a single-agent environment where the agent is allocated with an ex-ante allocation probability $q$,
we show that the ratio of the optimal surplus to the optimal consumer surplus is at most $1 - \ln q$.
Our reduction allows us to convert the $(1 - \ln q)$-approximation in the single-agent setting to a $(1 + o(1) \ln n$-approximation for any multi-agent environment.

We reduce designing multi-agent approximation mechanisms to single-agent approximations with ex-ante constraints $q$ bounded away from 0.
The approach is inspired by \citet{FHL20}, who prove a similar reduction when the single-agent approximation is bounded above by a constant for all $q \in [0, 1]$.
Their approach cannot be directly applied to the optimal omniscient approximation of consumer surplus because the single-agent approximation $(1 - \ln q)$ approaches $\infty$ as $q \xrightarrow{} 0$.
We adapt ideas from \citet{HT19} to prove a reduction to the single-agent setting with an ex-ante constraint $q \in [1/n, 1]$.
In this range, the single-agent approximation ratio is at most $(1 + \ln n)$.
The resulting construction establishes an optimal omniscient approximation of $(1 + o(1)) \ln n$ for multi-agent environments with arbitrary feasibility constraints.

The second step of the reduction framework from pen testing to deferred-acceptance auctions upper bounds the standard approximation $\apxda(n)$, the ratio between the optimal consumer surplus of a sealed-bid mechanism and the consumer surplus of a given deferred-acceptance auction.
Through Myerson's virtual value framework \citep{Myer81}, we argue that the above is equivalent to finding an upper bound for the optimal expected surplus divided by the expected surplus of some deferred-acceptance auction.
Surplus optimal deferred-acceptance auctions are known for $k$-identical goods and more general matroid feasibility constraints \citep{MS14, BdVSV11}. We give a $2$-approximate deferred-acceptance mechanism for knapsack constraints. Beyond matroids and knapsacks, \citet{FGGS22} give a logarithmic approximate deferred-acceptance auction for any downward-closed feasibility constraint. We observe an inherent downward-closure in the pen testing problem that allows the auction of \citet{FGGS22} to be extended to all combinatorial constraints. The set of pens chosen by any pen testing algorithm can be padded by adding (possibly expended) pens to always output a maximal feasible set.

Online posted-price mechanisms are a special case of deferred-acceptance auctions. Online environments are ones where an irrevocable decision regarding choosing a pen needs to be taken before testing the next one. We consider the settings where the order of testing pens is adversarially chosen (the oblivious version, \textit{cf.} prophet inequalities, \citealp{CHMS10}) and where the algorithm can choose the order (the sequential version, \textit{cf.} correlation gap, \citealp{Yan11}). With our analysis of the optimal omniscient approximation $\apxcs$ and known bounds on the standard approximation $\apxda$ of online pricing mechanisms, the reduction framework yields better performance guarantees than the $O(\log n)$ bounds by \citet{QV23} by constant factors for these environments.  

The results detailed above are summarized in \Cref{table:Summary}.

\begin{table}
\centering
\resizebox{0.9\textwidth}{!}
{\begin{tabular}{c c c c c c}
    \hline
    Environment & \multicolumn{2}{c}{$\frac{\apxcs(n)}{(1+o(1))}$} & $\apxda(n)$ & \multicolumn{2}{c}{$\frac{\apxom(n)}{(1+o(1))}$} \\
    \hline
    Select $k = o(n)$ & $\tfrac{2}{\ln 2} \, \ln \frac{n}{k}$\textsuperscript{*} & ${\ln \frac{n}{k}}$ & $1$\textsuperscript{**} & \multicolumn{2}{c}{${\ln \frac{n}{k}}$} \\
    Select $k = \Theta(n)$ & $\tfrac{2}{\ln 2} \, \ln \frac{n}{k}$\textsuperscript{*} & ${2.27 \, (0.577 + \ln \frac{n}{k})}$ & $1$\textsuperscript{**} & \multicolumn{2}{c}{${2.27 \, (0.577 + \ln \frac{n}{k})}$} \\
    Matroids & $\tfrac{2}{\ln 2} \, \ln m$\textsuperscript{$\S$} & ${\ln n}$ & $1$\textsuperscript{**} & \multicolumn{2}{c}{${\ln n}$} \\
    Knapsack & $\tfrac{2}{\ln 2} \, \ln m$\textsuperscript{$\S$} & ${\ln n}$ & $2$ & \multicolumn{2}{c}{${2 \cdot \ln n}$} \\
    General Downward-Closed & $\tfrac{2}{\ln 2} \, \ln m$\textsuperscript{$\S$} & ${\ln n}$ & $O(\log \log m)$\textsuperscript{$\dagger$} & \multicolumn{2}{c}{${O(\log n \log 
    \log m)}$} \\
    \hline
    Select $1$ Online Oblivious &  \multicolumn{2}{c}{${\ln n}$} & $2$\textsuperscript{$\dagger \dagger$} & $O(\log n)$\,\textsuperscript{$\ddagger$} & ${2 \cdot \ln n}$ \\
    Select $1$ Online Sequential & \multicolumn{2}{c}{${\ln n}$} & $\frac{e}{e-1}$\textsuperscript{$\ddagger \ddagger$} & $O(\log n)$\,\textsuperscript{$\ddagger$} & ${\frac{e}{e-1} \cdot \ln n}$ \\
    Select $1$ Online IID & \multicolumn{2}{c}{${\ln n}$} & $\frac{e}{e-1}$\textsuperscript{$\ddagger \ddagger$} & $e \ln n$\,\textsuperscript{$\ddagger$} & ${\ln n}$ \\
    Matroid Online Oblivious &  \multicolumn{2}{c}{${\ln n}$} & $2$\textsuperscript{$\S\S$} & \multicolumn{2}{c}{${2 \cdot \ln n}$} \\
    Matroid Online Sequential & \multicolumn{2}{c}{${\ln n}$} & $\frac{e}{e-1}$\textsuperscript{$\ddagger \ddagger$} &  \multicolumn{2}{c}{${\frac{e}{e-1} \cdot \ln n}$} \\
    \hline
\end{tabular} }
\caption{Prior bounds (if any, references indexed by superscripts) and our upper bounds. $n = \text{number of agents}$, $m = \text{number} \allowbreak \text{ of} \allowbreak \text{ feasible} \allowbreak \text{ outcomes}$. Note that $m \geq n$ for all the environments listed in the table. The bounds for optimal omniscient approximation $\apxcs$ and omniscient approximation $\apxom$ are normalized by a $(1+o(1))$ factor. Prior work: \textsuperscript{*}~\citealp{HR08}; \textsuperscript{**}~\citealp{MS14}, \citealp{BdVSV11}; \textsuperscript{$\S$}~\citealp{FTT16}; \textsuperscript{$\dagger$}~\citealp{FGGS22}; \textsuperscript{$\dagger \dagger$}~\citealp{ESC84}, \citealp{CHMS10}; \textsuperscript{$\S\S$}~\citealp{KW12}; \textsuperscript{$\ddagger$}~\citealp{QV23}; \textsuperscript{$\ddagger \ddagger$}~\citealp{CHMS10}, \citealp{Yan11}. All bounds in the table also hold for the underlying auction environment. Pen testing algorithms with the same guarantees as downward-closed constraints can be achieved for arbitrary combinatorial constraints.}
\label{table:Summary}
\end{table}

Our reduction from pen testing to deferred-acceptance auction design improves bounds from previous literature for all the environments that we discuss;
however, it does not always yield a tight bound. For the online IID environment (the online environment where the ink levels of all pens are drawn IID), we identify an algorithm with an omniscient approximation $\apxom(n)$ that is a constant factor less than $\apxda(n) \, \apxcs(n)$.

Our results and techniques are of interest beyond pen testing.
Our upper bound of $(1 + o(1)) \ln n$ on the optimal omniscient approximation $\apxcs(n)$ is the ratio between the optimal surplus and the optimal consumer surplus, and is a mechanism design problem of independent interest \citep{HR08, FTT16, GL24}.
Recall that the surplus created in an auction equals the sum of the consumer surplus and the revenue generated by the auctioneer. However, when payments made by the winners are burnt instead of being transferred to the auctioneer, the total surplus created is just the consumer surplus. The need to burn payments can be diverse. For instance:
\begin{enumerate}
    \item Making monetary transactions is not desirable or repugnant, like in setting prices for welfare schemes for vulnerable populations. Ordeal-based mechanisms are employed in such scenarios, where agents make non-transferable non-monetary payments, like spending time waiting in a line to avail the provisions provided in the welfare scheme \citep{KK11, ABHOPW13, BCK16}.
    \item In computer systems like data packet routing networks, where technological infeasibility to collect payments necessitates charging non-monetary payments \citep{DN92, DGN03, HR08}.
    \item The auctioneer cannot be trusted to truthfully implement the auction, like in a blockchain, where a miner runs an auction to decide which set of transaction goes into the next block it proposes. To stop the miner from colluding with agents or from benefiting from injecting fake bids, it becomes essential to burn a substantial portion of the transaction fees collected from the agents \citep{Rou21, CS23, GTW24}.
\end{enumerate}

The technique to obtain the $(1 - \ln q)$-approximation in the single-agent environment can have further applications.
We express consumer surplus optimization as a linear program and explicitly compute the worst-case distribution.
This provides a proof technique which automatically gives tight approximation guarantees.
For example, we reuse the technique to compute the worst-case distribution for the omniscient approximation for the online IID environment.
We improve the bound obtained from the reduction framework to a multiplicative approximation ratio within an additive 1 of the $H_n$ lower bound of \citet{QV23} (where $H_n$ is the $n$\textsuperscript{th} harmonic number).

\textbf{Paper Organization.}
We make the reduction from the pen-testing problem to designing deferred-acceptance auctions formal in \Cref{sec:PTtoDAA}.
We execute the first step of the reduction
framework — upper bounding the optimal omniscient approximation $\apxcs(n)$ — in \Cref{sec:CSvS}.
We bound the approximation ratio between the surplus and consumer surplus in a single-agent environment in \Cref{sec:SingleAgent} and discuss the reduction from multi-agent to single-agent settings in \Cref{sec:Multi}.
The second step of the reduction from pen testing algorithms to deferred-acceptance mechanisms is discussed in \Cref{sec:DAAPrev}, where we review deferred-acceptance mechanisms for various environments and discuss the corresponding bounds we get for pen testing.
Finally, we revisit the online IID single item environment and achieve a tight (up
to $o(1)$ terms) bound in \Cref{sec:IID}.

\section{The Pen Testing Problem and Ascending-Price Auctions} \label{sec:PTtoDAA}
In this section, we formally define the combinatorial pen testing problem, generalized from \citet{QV23}, and explore its connections to ascending-price mechanisms.

\begin{comment}
\begin{table}
\centering
\resizebox{1.1\textwidth}{!}
{\begin{tabular}{c c c c c c}
    \hline
    Environment & \multicolumn{2}{c}{$\apxcs(n)$} & $\apxda(n)$ & \multicolumn{2}{c}{$\apxom(n)$} \\
    \hline
    $k$-Identical Goods & $2 \ln 2 \cdot (1 + o(1)) \ln \frac{n}{k}$ \textsuperscript{a}\, & $(1+o(1)) \ln \frac{n}{k}$ & $1$\textsuperscript{b} & \multicolumn{2}{c}{$(1+o(1)) \ln n$} \\ 
    Matroids &  \multicolumn{2}{c}{$(1 + o(1)) \ln \lceil \max_{\emptyset \neq X} \frac{|X|}{rank(X)} \rceil$} & $1$\textsuperscript{b} & \multicolumn{2}{c}{$(1 + o(1)) \ln \lceil \max_{\emptyset \neq X} \frac{|X|}{rank(X)} \rceil$} \\
    General Downward-Closed & \multicolumn{2}{c}{$(1 + o(1)) \ln n$} & $O(\log n)$\textsuperscript{c} & \multicolumn{2}{c}{$O(\log^2 n)$} \\
    \hline
    Online Prophet Inequalities &  \multicolumn{2}{c}{\multirow{3}{*}{$(1+o(1)) \ln n$}} & $2$\textsuperscript{d} & \multirow{3}{*}{$O(\log n)\,$\textsuperscript{e}} & $2(1+o(1)) \ln n$ \\
    Online Correlation Gap & & & $\frac{e}{e-1}$\textsuperscript{f} & & $\frac{e}{e-1} (1+o(1)) \ln n$ \\
    Online IID & & & $\frac{e}{e-1}$\textsuperscript{f} & & $(1 + o(1)) \ln n$ \\
    \hline
\end{tabular} }
\caption{Summary of Our Results (prior work-- {\footnotesize \textsuperscript{a} \citealp{HR08}; \textsuperscript{b} \citealp{MS14}; \textsuperscript{c} \citealp{FGGS22}; \textsuperscript{d} \citealp{CHMS10}, \citealp{ESC84}; \textsuperscript{e} \citealp{QV23}; \textsuperscript{f} \citealp{CHMS10}, \citealp{Yan11}})}
\label{table:Summary}
\end{table}
\end{comment}

\begin{definition} \label{def:pentesting}
    A pen testing instance is described by a set $N = \{1, 2, \dots, n\}$ of pens with unknown ink levels $\val_1, \dots, \val_n$, each $\val_i$ drawn independently from the distribution $F_i$, and a subset $\feasibility$ of the powerset of the pens denoting feasible subsets of pens. The residual ink level $\util_i$ of pen $i$ is initiated to $\val_i$. We can test the pens before making a decision. Test $(i, \theta_i)$ is done by writing with pen $i$ for time $\theta_i$. We receive a binary signal at the end of the test:
        \begin{enumerate}
            \item If $\util_i \geq \theta_i$: The test succeeds. The pen now has an (unobservable) ink level $\util_i \xleftarrow{} \util_i - \theta_i$.
            \item If $\util_i < \theta_i$: The test fails. The pen now has no ink left, i.e, $\util_i \xleftarrow{} 0$.
        \end{enumerate}
        We can choose to run multiple tests on the pens. We need to output some feasible subset $P \in \feasibility$ of pens, maximizing
        $$\sum\nolimits_{i \in P} \util_i.$$
\end{definition}
The goal is to optimize the total amount of ink remaining in the set of chosen pens.

Compare a pen testing instance with $n$ pens and a feasibility constraint $\feasibility$ with an $n$-agent auction under the same feasibility constraint. The original levels of ink correspond to the value each agent gets upon being allocated. The ink spent through testing is analogous to the prices in the auction. Thus, maximizing total residual ink is equivalent to optimizing consumer surplus, i.e, the sum of the values of the winning agents minus the sum of all payments.
Note that auctions generally give the auctioneer the added advantage of the ability to solicit bids from agents. In pen testing algorithms, we only get to learn whether the ink left in the pen is more than some threshold and in doing so, we irrevocably expend ink up to the threshold. This is similar to ascending-price auctions, where the auctioneer irreversibly increases the price faced by each agent, and in doing so, only learns whether the agent has a value at least the price.

\citet{MS14} describe the wide class of ascending-price auctions called deferred-acceptance mechanisms.

\begin{definition}[\citealp{MS14}; Deferred-Acceptance Auctions] \label{def:DA}
     A deferred-acceptance auction is held across multiple stages. For each stage $t$, the auction maintains a set of active bidders $A_t$ satisfying $A_1 \supseteq A_{2} \supseteq \dots \supseteq A_t$, where the initial active set $A_1$ is the set of all bidders. The auction is characterized by a (possibly randomized) pricing rule $\vec{p}$, mapping the history of the auction at stage $t$ to (discriminatory) prices for each agent satisfying $p_i(H_t) \geq p_i(H_{t-1})$ for all agents $i$ and for all histories $H_t$ at stage $t$ arising out of stage $t-1$ history $H_{t-1}$, i.e, the prices are monotonously non-decreasing for all agents over time. Agents can opt to drop out once the prices are updated in stage $t$. $A_{t+1}$ is the set of agents in $A_t$ that did not drop out in stage $t$. The mechanism terminates at stage $t$ when $A_t$ becomes feasible. The agents in $A_t$ are charged according to $\vec{p}(H_t)$.
\end{definition}
The definition from \citet{MS14} restricts deferred-acceptance mechanisms to be deterministic, which we relax for the purpose of designing pen testing algorithms. From the discussion above, note that any deferred-acceptance mechanism can be converted into a pen testing algorithm, by setting thresholds equal to the prices recommended by the mechanism.
Thus, the optimal (near optimal) pen testing algorithm corresponds to the consumer surplus optimal (near optimal) deferred-acceptance auction.
%The consumer surplus of the mechanism corresponds to the residual ink in the pens chosen by the pen testing algorithm.

We compare the performance of our deferred-acceptance mechanisms (and the corresponding pen testing algorithms) to the following benchmarks:
\begin{itemize}
    \item The \emph{omniscient benchmark} \citep{QV23}: The omniscient benchmark corresponds to the omniscient pen testing algorithm that exactly knows the ink levels in the pens, and thus, does not have to waste ink by writing with them.
    In the reduction to the auction environment, the omniscient benchmark corresponds to a mechanism that neither loses out on the payments nor has to stick to an increasing trajectory of prices.
    Thus, the omniscient benchmark corresponds to the surplus optimal sealed-bid mechanism.
    The \emph{omniscient approximation} $\apxom(n)$ corresponds to the ratio between the omniscient benchmark and the performance of a given pen testing algorithm in the pen testing setting, and the omniscient benchmark and the consumer surplus of a given deferred-acceptance auction in the auction environment.
    \item The \emph{standard benchmark}: The standard benchmark corresponds to the consumer-surplus-optimal sealed-bid mechanism.
    We define the \emph{standard approximation} $\apxda(n)$ to be the ratio between the standard benchmark and the consumer surplus of a candidate deferred-acceptance mechanism.
    While the standard benchmark does not have a natural pen testing analogy, observe that the standard benchmark is a tighter upper bound to the optimal pen testing algorithm than the omniscient benchmark.
    Thus, the standard approximation gives a better upper bound on the ratio between the optimal pen testing algorithm and the residual ink gathered by a given pen testing algorithm.
\end{itemize}

Now that we have established the consumer surplus is a quantity of interest for pen testing, we review consumer surplus optimization through virtual values in the next section.

\subsection{Consumer Surplus Optimization and Virtual Valuations}
Let us begin in a single-agent environment with one good. Let the agent's value be drawn from the distribution $F$. The quantile $q$ of an agent with value $\val \sim F$ is the measure with respect to $F$ of stronger values, i.e, $q = 1 - F(v)$. For $q \in [0, 1]$, let $\val(q)$ be the inverse demand function of $F$ satisfying $F(\val(q)) = 1-q$. In other words, $Pr_{\hat{v} \sim F}(\hat{v} > \val(q)) = q$. Throughout the paper, we assume $\val(1) = 0$ and $\int_0^0 \val(t)dt = 0$. It can be shown that these assumptions can be made without loss of generality.

Let $\cumval(q) = \int_0^q \val(t)dt$ be the price-posting surplus curve. Notice that $\cumval(q)$ is the expected surplus from posting a price $\val(q)$ to the agent, thereby allocating the good with probability $q$. By our assumption on $F$, $\cumval(0) = 0$. Let $\cumutil(q) = \cumval(q) - q\val(q)$ be the price-posting consumer surplus curve. Similar to the price-posting surplus curve, $\cumutil(q)$ is the consumer-surplus from posting a price $\val(q)$. Note that $0 \leq \cumutil(0) \leq \cumval(0) = 0$, and thus, $\cumutil(0) = 0$.

Let $\util(q) = U'(q) = -q\val'(q)$ be the marginal price-posting consumer surplus curve. Note that $\val$ is monotonously non-increasing, and hence its derivative $\val'$ is non-positive. Thus, $\util(q) \geq 0$ and $\cumutil$ is monotone non-decreasing. A quick note on derivatives: whenever $\val^\prime(q)$ is well defined, $\util(q)$ is well defined. At other points, $\util(q)$ can be calculated using any sub-derivative of $\val$. At points $q$ of discontinuity of $\val$, $\val'(q) = -\infty$ and thus $\util(q) = \infty$.

\begin{theorem}[\citealp{Myer81}] \label{thm:VirtualWelfare}
     In a Bayesian incentive-compatible mechanism with allocation rule $y$ and payment rule $p$ (over quantiles), the expected consumer surplus of an agent satisfies $$E_{q \sim U[0,1]}[\val(q)y(q) - p(q)] = E_{q \sim U[0,1]}[\util(q)y(q)] = E_{q \sim U[0, 1]}[-\cumutil(q) \, y'(q)] + U(0)y(0)$$ 
\end{theorem}
Thus, the expected consumer surplus equals the expected virtual surplus with marginal consumer surplus as virtual values. \citet{Myer81} states that an allocation rule can be implemented as a truthful auction if and only if the allocation rule is monotone non-increasing in quantile space. Thus, the consumer surplus optimal mechanism optimizes for virtual surplus subject to the allocation rule being monotonously non-increasing.

When the marginal price-posting consumer surplus curve is non-increasing (and hence $\cumutil$ is concave), optimizing for virtual surplus automatically ensures monotonicity of the allocation rule. We will call these distributions with non-increasing marginal price-posting consumer surplus curves \emph{consumer surplus regular}.

However, $u$ is not monotone non-increasing for common distributions like the uniform and normal distribution. In fact, $u$ is monotone non-decreasing for these distributions. In such a case, the allocation that pointwise optimizes for virtual surplus might not satisfy monotonicity. \citet{Myer81} prescribes an ironing procedure to optimize for virtual surplus subject to monotonicity.
\begin{enumerate}
    \item Construct the concave hull $\icumutil$ of the price-posting consumer surplus curve $\cumutil$
    \item Define $\iutil(q) = \icumutil'(q)$
    \item Find the virtual surplus optimal mechanism using $\iutil$ as virtual surplus
\end{enumerate}
Since $\icumutil$ is concave, $\iutil$ is non-increasing, and hence optimizing for the ironed virtual surplus is easier. Throughout the paper, we will follow the convention of attaching a bar on top of a curve to describe the ironed curve.

In a multi-agent environment, the interim allocation rule for an agent is the single-agent allocation rule that arises in expectation over the quantiles of all other agents. This captures the perspective of the agent after knowing its value, but before learning the values of the other agents.
\begin{theorem} \label{thm:EquivIroning}
    Let $y$ be the interim allocation rule for some agent, satisfying $\frac{d}{dq} y(q) = 0$ for all $q$ such that $\cumutil(q) \neq \icumutil(q)$. Then
    $$E_{q}[\util(q)y(q)] = E_{q}[\iutil(q)y(q)]$$
    In other words, if the allocation rule does not differentiate between agents in an ironed region, then the ironed virtual surplus can be used to compute the consumer surplus instead of the actual virtual surplus.
\end{theorem}

\begin{theorem}[\citealp{AFHHM12,AFHHM19}] \label{thm:Amortize}
    Let $y$ be some allocation rule with interim allocation rule $y_i$ monotonously non-increasing for each agent $i$. Then, there exists a mechanism $\overline{y}$ with interim allocation rule $\overline{y}_i$ for agent $i$ such that the expected consumer surplus for agent $i$ equals
    $$E_{q \sim U[0, 1]}[\util_i(q) \, \overline{y}_i(q)] = E_{q \sim U[0, 1]}[\iutil_i(q) \, y_i(q)]$$
    In other words, the expected consumer surplus of $\overline{y}$ is the consumer surplus of the original mechanism if the virtual values were given by $\iutil_i$ for agent $i$ instead of $\util_i$
\end{theorem}

\subsection{Near-Optimal Deferred-Acceptance Mechanisms for Consumer Surplus}
In this section, we prescribe a recipe similar to \citet{FHL20} to convert any approximately optimal deferred-acceptance mechanism for surplus into an approximately optimal deferred-acceptance mechanism for consumer surplus.

\begin{definition}[Virtual-Pricing Transformation of a Deferred-Acceptance Mechanism] \label{def:VirtualWelfareImplementation}
     Let $\DAOM$ be a deferred-acceptance mechanism designed to optimize surplus. Let $\val_i$ be the inverse demand function and $\util_i$ be the virtual value function (for consumer surplus) for agent $i$. The virtual-pricing transformation on $\DAOM$, denoted by $\DACS$, implements $\DAOM$ in ironed virtual price space, i.e, whenever $\DAOM$ posts a price $\hat{v}_i$ to agent $i$, $\DACS$ posts a price $\val_i(\theta_i)$ satisfying
    $$\theta_i = \sup \{\theta: \iutil_i(\theta) \geq \hat{v}_i\}$$
\end{definition}

We will make some preliminary observations about the transformation:
\begin{enumerate}
    \item Consider two runs of the transformed mechanism $\DACS$ with some agent having two different values, both corresponding to the same ironed virtual value (which can happen if both these values are within the same ironed interval). From the definition of the transform, the agent faces the smallest price needed to differentiate itself from virtual values smaller than the threshold set by the mechanism. Consequently, the mechanism does not post prices from the middle of an ironed interval. Thus, in both these runs of the mechanism, the agent behaves identically. Since the mechanism does not discriminate between values with the same ironed virtual value, the consumer surplus of the mechanism equals the ironed virtual surplus.
    \item If the original mechanism $\DAOM$ is a $\apxda$-approximation to the optimal surplus, in expectation over all product distributions, then the transformed mechanism $\DACS$ is a $\apxda$-approximation to the optimal ironed virtual surplus, and hence, to the optimal consumer surplus.
    \item Finally, it is straightforward to see that the transformed mechanism $\DACS$ posts non-decreasing prices to each agent. Hence, $\DACS$ is also a deferred-acceptance mechanism.
\end{enumerate}
Recall that the standard approximation $\apxda(n)$ is the ratio between the consumer surplus of the optimal deferred-acceptance mechanism and our proposed mechanism while the optimal omniscient approximation $\apxcs(n)$ is the ratio between the optimal surplus and the optimal consumer surplus. The proof of \Cref{thm:IntroSurplusDA} is now straightforward.

\begin{proof}[Proof of \Cref{thm:IntroSurplusDA}]
    Let $\OPTCS$ denote the expected optimal consumer surplus, $\OPTOM$ denote the expected optimal surplus, and $\DACS$ denote the expected consumer surplus of the deferred acceptance mechanism in \Cref{def:VirtualWelfareImplementation}. Then,
    $$\DACS \geq \tfrac{\OPTCS}{\apxda(n)} \geq \tfrac{\OPTOM}{\apxda(n) \, \apxcs(n)}$$
    where the first inequality follows from the definition of standard approximation $\apxda(n)$ and bullet $2$ in the discussion above, while the second follows from the definition of optimal omniscient approximation $\apxcs(n)$. The equivalent pen testing algorithm yields the necessary approximation ratios against the two benchmarks.
\end{proof}

In the next section, we analyze the optimal omniscient approximation $\apxcs(n)$ for general combinatorial environments.

\section{Consumer Surplus versus Surplus} \label{sec:CSvS}
Our approximation procedure follows two steps. In \Cref{sec:SingleAgent}, we show that the optimal surplus is at most $1-\ln q$ times the consumer surplus in a single-agent environment where the ex-ante probability of allocation is constrained to be at most $q$. In \Cref{sec:Multi}, we move from the single-agent environment to multi-agent environments combining approaches from \citet{FHL20} and \citet{HT19} to show a $(1+o(1)) \ln n$ approximation between the consumer surplus and optimal surplus in $n$-agent environments with general feasibility constraints.

%However, the single-agent approximation diverges to $\infty$ when $q \xrightarrow{} 0$, and does not help in approximating the consumer surplus against the optimal surplus for these small quantiles. \Cref{sec:Multi} applies an approach of \citet{HT19}, that shows it is sufficient to have a good approximation for consumer surplus only in the larger quantiles ($q = \omega(\frac{1}{n})$), to get around this challenge and to give a $(1+o(1)) \ln n$ approximation for $n$-agent environments with general feasibility constraints.

\subsection{The Single-Agent Problem} \label{sec:SingleAgent}
Recall that $\cumval, \cumutil$ and $\icumutil$ are the price-posting surplus curve, price-posting consumer surplus curve and the ironed consumer surplus curve respectively. In this section, we establish a ``closeness property" like in \citet{FHL20}.
\begin{theorem} \label{thm:Main}
    For an ex-ante allocation probability $q \in [0, 1]$, the ratio of the optimal surplus $\cumval(q)$ to the optimal consumer surplus $\icumutil(q)$ is at most $1 - \ln q$.
\end{theorem}
We will begin by proving \Cref{thm:Main} for consumer surplus regular distributions, and then extend the result to all distributions.

\begin{lemma} \label{thm:RSRDist}
    For any consumer surplus regular distribution and an ex-ante allocation probability $q \in [0, 1]$, the ratio of the optimal surplus $\cumval(q)$ to the optimal consumer surplus $\cumutil(q)$ is at most $1 - \ln q$.
\end{lemma}
\begin{proof}
Conditional on generating a surplus $\cumval(q)$, we want to compute the distribution that minimizes consumer surplus $\cumutil(q)$.
\begin{align}
\cumutil(q) &= \int_0^q \util(t)dt. \nonumber \\
        \cumval(q) &= \cumutil(q) + q\val(q) \nonumber\\
        &= \int_0^q \util(t)dt - q\int_q^1 \val'(t)dt \nonumber\\
        \label{eqn:FirstVEqn}
        &= \int_0^q \util(t)dt + q\int_q^1 \frac{\util(t)}{t}dt.
\end{align}
We substitute $\util(t) = -t\val'(t)$ in the last equation.

Simultaneously, we enforce $\util(t) \geq 0$, and $\util(t)$ is non-increasing (since the distribution is consumer surplus regular), $\val(1) = 0$ and $\cumval(0) = 0$.

Finding the minimum consumer surplus distribution reduces to the following program.
$$\min \int_0^q \util(t)dt$$
subject to
$$\int_0^q \util(t)dt + q\int_q^1 \frac{\util(t)}{t}dt = \cumval(q)$$
$$\util(t) \geq 0, \util(t) \text{ is monotonously non-increasing}$$
$$\val(1) = 0, \cumval(0) = 0$$
The constraints in the last line can be enforced after solving for $u$.

Consumer surplus regularity dictates $\util(t) \geq \util(q)$ for $t \leq q$. Hence, minimizing $\int_0^q \util(t)dt$ would correspond to setting $\util(t) = \util(q)$ for $t \leq q$. Similarly, $\util(t) \leq \util(q)$ for $t \geq q$. Minimizing $\int_0^q \util(t)dt$ would mean minimizing $\cumval(q) - q\int_q^1 \frac{\util(t)}{t}dt$ which is achieved by setting $\util(t) = \util(q)$ for $t \geq q$. Let $u(0) = u(t) = u(1) = u$.
For a constant marginal price-posting consumer surplus curve,
$$\cumutil(q) = \int_0^q \util(t)dt = q \, u$$
$$\cumval(q) = \int_0^q \util(t)dt + q\int_q^1 \frac{\util(t)}{t}dt = [q - q\ln q] \, u$$
Thus, for any consumer surplus regular distribution, $\frac{\cumval(q)}{\cumutil(q)} \leq [1 - \ln q]$. The constant marginal consumer surplus curve that achieves this ratio corresponds to the exponential distribution.
\end{proof}

\begin{lemma} \label{thm:IncRSCurve}
    Let the consumer surplus curves $\cumutil, \hat{U}$ satisfy $\hat{U}(0) = 0$ and $\hat{U}(t) \geq \cumutil(t)$ for all $t \in [0, 1]$. Then, for the corresponding surplus curves, $\hat{V}(q) \geq \cumval(q)$.
\end{lemma}
\begin{proof}
We rewrite equation \eqref{eqn:FirstVEqn} as follow.
\begin{equation}
    \begin{split}
        \cumval(q) &= \int_0^q \util(t)dt + q\int_q^1 \frac{\util(t)}{t}dt \\
        &= \cumutil(q) + q \Bigl[\frac{\cumutil(t)}{t}\Bigr]_{t = q}^{t = 1} + q\int_q^1 \frac{\cumutil(t)}{t^2}dt \\
        &= q\cumutil(1) + q\int_q^1 \frac{\cumutil(t)}{t^2}dt
    \end{split}
\end{equation}
The second equality is obtained through integrating by parts. Increasing $\cumutil$ pointwise clearly results in an increase in $\cumval$.
\end{proof}

\begin{proof}[Proof of \Cref{thm:Main}]
The optimal consumer surplus for an ex-ante allocation probability $q$ equals $\icumutil(q)$.

Consider the distribution with price-posting consumer surplus curve $\icumutil$. Let $\icumval$ be the price-posting surplus curve corresponding to $\icumutil$. From \Cref{thm:IncRSCurve} and \Cref{thm:RSRDist},
$$\frac{\cumval(q)}{\icumutil(q)} \leq \frac{\icumval(q)}{\icumutil(q)} \leq 1 - \ln q$$
The second inequality holds from \Cref{thm:RSRDist}, since $\icumval$ is the price-posting surplus curve of a consumer surplus regular distribution (we made $\icumutil$ concave through ironing).
\end{proof}

See \Cref{sec:SurplusofOPTCS} for a brief discussion on the surplus generated by the single-agent optimal consumer surplus auction. 

\subsection{Multi-Agent Environments} \label{sec:Multi}
In this section, we convert the single-agent bound from \Cref{thm:Main} into an upper bound on the ratio between the optimal surplus and optimal consumer surplus (aka the optimal omniscient approximation) for multi-agent environments. We outline our approach before delving into the proof. Consider the interim allocation rule $y_i$ for agent $i$ in the surplus optimal mechanism $\OPTOM$. Without loss of generality, assume $y_i(1) = 0$. The expected surplus for agent $i$ equals
$$E_{q \in U[0, 1]}[y_i(q) \, \val_i(q)] = \bigl[y_i(q) \, \cumval_i(q) \bigr]^{q = 1}_{q = 0} + E_{q \in U[0, 1]}[-y_i'(q) \, \cumval(q)] = E_{q \in U[0, 1]}[-y_i'(q) \, V(q)]$$
This equality comes from integrating by parts ($y_i(q) \cumval_i(q) $ vanishes at both, $q = 0$ and $1$, since $\cumval_i(0) = 0$ and $y_i(1) = 0$). Suppose there exists a constant $\apxrough$ such that $\cumval_i(q) \leq \apxrough \, \icumutil_i(q)$ for all $q \in [0, 1]$, then,
$$E_{q \in U[0, 1]}[-y_i'(q) \, \cumval(q)] \leq E_{q \in U[0, 1]}[-y_i'(q) \, \alpha \, \icumutil(q)] = \apxrough \cdot E_{q \in U[0, 1]}[y_i(q) \, \iutil_i(q)]$$

The above inequality connects the expected surplus to the expected ironed consumer surplus of the interim allocation rule $y_i$. Since $\frac{d}{dq} y_i$ need not be equal to $0$ whenever $\cumutil(q) \neq \icumutil(q)$, $E_{q \in U[0, 1]}[y_i(q) \, \iutil_i(q)] \neq E_{q \in U[0, 1]}[y_i(q) \, \util_i(q)]$. But, \Cref{thm:Amortize} shows the existence of a mechanism $\overline{y}$ with expected consumer surplus of each agent $i$ equal to $E_{q \in U[0, 1]}[y_i(q) \, \iutil_i(q)]$. Thus, the expected consumer surplus of $\overline{y}$ is an $\apxrough$-approximation to the optimal surplus. The consumer surplus optimal mechanism $\OPTCS$ will only do better. 

%\Anote{We know that changing the interim rule from $y_i$ to $\overline{y_i}$ will get the correct consumer surplus for agent $i$. But, it is not obvious why the interim allocation rules $\overline{y_1}, \dots, \overline{y_n}$ can together form a mechanism. That's why we are citing the result for ALaei et al.}

Unfortunately, as $q \xrightarrow{} 0$, the bound from \Cref{thm:Main} of $\frac{\cumval_i(q)}{\icumutil_i(q)} \leq (1 - \ln q)$ diverges to $\infty$. Thus, we need to show that the loss from ignoring the small (strong) quantiles is not much. Suppose we find a mechanism $\bridgemechanism{1}$ with interim allocation $\bridgeallocation{i}$ for agent $i$, near optimal surplus, and $\bridgeallocation{i}'(q) = 0$ for $q \in [0, \epsilon]$. Then, $y_i'(q) \times \cumval(q) = y_i'(q) \times \cumutil(q) = 0$ in this region. For $q \in [\epsilon, 1]$, we know $\cumval(q) \leq (1 - \ln \epsilon) \times \icumutil(q)$. Thus,
$$(1 - \ln \epsilon) \times \OPTCS \geq \text{surplus}(\bridgemechanism{1})$$

\citet{HT19} give the construction of such a mechanism $\bridgemechanism{1}$, originally intended to design good mechanisms with only sample access to the distributions of bidders' bids.
\begin{definition}[\citealp{HT19}; $\epsilon$-buffering rule]
    Given an allocation rule $\vec{y} = (y_1, y_2, \dots, y_n)$ and a quantile $\epsilon \in [0, 1]$, the $\epsilon$-buffering rule for $\Vec{y}$ simulates $y$ with quantiles transformed on each agent as follows:
    \begin{itemize}
        \item Top inflate: for any $q_i \in [0, \epsilon]$, return $0$
        \item For any $q_i \in [\epsilon, 1-\epsilon]$, return $\frac{q_i - \epsilon}{1-2\epsilon}$
        \item Bottom deflate: for any $q_i \in [1-\epsilon, 1]$, return $1$
    \end{itemize}
\end{definition}
In essence, the top inflate branch makes the $\epsilon$-buffering rule treat all agents with a small quantile as if they had quantile $0$ and thus, $y_i'(q) = 0$ for $q \in [0, \epsilon]$ (the bottom deflate branch is unnecessary for the analysis of monotone payoff curves like surplus $\cumval$ and consumer surplus $\cumutil$; see \Cref{rem:BottomDeflate}).

\begin{theorem}[\citealp{HT19}] \label{thm:HTAPX}
    Let $\bridgemechanism{1}$ be the $\epsilon$-buffering rule of the surplus optimal mechanism $\OPTOM$. Then, $\OPTOM \leq \frac{1}{(1 - \frac{\epsilon}{1 - \epsilon}) \, (1 - \epsilon) \, (1 - 2n \epsilon)} \times \text{surplus}(\bridgemechanism{1})$.
\end{theorem}

\begin{theorem} \label{thm:GeneralFeasibleMultiAgentEnvironments}
    Consider an $n$-agent environment with an arbitrary feasibility constraint. The optimal omniscient approximation is at most $\frac{1 - \ln \epsilon}{(1 - \frac{\epsilon}{1 - \epsilon}) \, (1 - \epsilon) \, (1 - 2n \epsilon)}$.
\end{theorem}
\begin{proof}
    The proof follows by combining \Cref{thm:HTAPX} with the discussion above. Let $\bridgemechanism{1}$ be the $\epsilon$-buffering rule of the surplus optimal mechanism $\OPTOM$. Then,
    $$\OPTOM \leq \frac{1}{(1 - \frac{\epsilon}{1 - \epsilon}) \, (1 - \epsilon) \, (1 - 2n \epsilon)} \times \text{surplus}(\bridgemechanism{1}) \leq \frac{1 - \ln \epsilon}{(1 - \frac{\epsilon}{1 - \epsilon}) \, (1 - \epsilon) \, (1 - 2n \epsilon)} \times \OPTCS$$
\end{proof}

\begin{corollary} \label{thm:Tight}
    In $n$-agent environments with an arbitrary feasibility constraint, the optimal omniscient approximation $\apxcs(n) \leq (1 + o(1)) \ln n$.
\end{corollary}
\begin{proof}
    Setting $\epsilon = \frac{1}{n \ln n}$ in \Cref{thm:GeneralFeasibleMultiAgentEnvironments}, we get
    \begin{equation}
        \notag
        \begin{split}
            \zeta(n) &\leq \frac{1 - \ln \frac{1}{n \ln n}}{(1 - \frac{\frac{1}{n \ln n}}{1 - \frac{1}{n \ln n}}) \, (1 - \frac{1}{n \ln n}) \, (1 - 2n \cdot \frac{1}{n \ln n})} \\
            &= \frac{1 + \ln n + \ln \ln n}{(1 - \frac{1}{n \ln n - 1}) \, (1 - \frac{1}{n \ln n}) \, (1 - \frac{2}{\ln n})} \\
            &= \frac{n \ln n - 1}{n \ln n - 2} \cdot \frac{n \ln n}{n \ln n - 1} \cdot \frac{\ln n - 2}{\ln n} \cdot \frac{1 + \ln n + \ln \ln n}{\ln n} \times \ln n \\
            &= (1 + o(1)) \ln n \qedhere
        \end{split}
    \end{equation}
\end{proof}

\begin{remark} \label{rem:BottomDeflate}
    The $\epsilon$-buffering rule of \citet{HT19} can be implemented without a bottom deflate branch. This would strengthen the approximation guarantee to $\frac{1}{(1-\epsilon) \, (1-n \epsilon)}$. However, substituting $\epsilon= \frac{1}{n \ln n}$ will not result in a bound tighter than $(1+o(1)) \ln n$.
\end{remark}

\begin{remark} \label{rem:FTT16}
    \citet{FTT16} show that the optimal omniscient approximation as a function of the number of outcomes $m$ in the feasibility constraint is at most $\tfrac{2}{\ln 2} (1 + o(1)) \ln m$. \Cref{thm:Tight} improves their bound whenever the feasibility constraint has cardinality at least $\sqrt{n}$. Note that all downward-closed constraints have at least $n$ feasible outcomes. For such feasibility constraints, \Cref{thm:Tight} yields an exponential improvement. As the number of feasible outcomes become exponential in $n$, their bound declines to $O(n)$.
\end{remark}

\begin{remark} \label{rem:Tight}
    The upper bound from \Cref{thm:Tight} matches the lower bound of \citet{QV23}. In the single-item environment with values of agents drawn IID from the exponential distribution, they show that the optimal surplus is $H_n = (1+o(1)) \ln n$ times the optimal consumer surplus, where $H_n$ is the $n^{th}$ harmonic number. 
\end{remark}

\subsection{The $k$-Identical Goods Environment}
\Cref{thm:GeneralFeasibleMultiAgentEnvironments} and \Cref{thm:Tight} give an over-arching bound for the optimal omniscient approximation $\apxcs(n)$ that holds in any environment. However, specific environments admit better approximations. For instance, we will give a tighter bound for $\apxcs(n)$ for the $k$-identical goods environment, where any set of at most $k$-agents can be allocated simultaneously.

\begin{theorem} \label{thm:TractableTightkID}
    In the $n$-agent $k$-identical goods environment, the optimal omniscient approximation satisfies
    \begin{enumerate}
        \item $\apxcs(n) \leq (1 + o(1)) \ln \frac{n}{k}$ when $k = o(n)$, and
        \item $\apxcs(n) \leq 2.27 \, (1 + o(1)) (0.577 + \ln \tfrac{n}{k})$, when $k = \Theta(n)$.
    \end{enumerate}
\end{theorem}
To prove \Cref{thm:TractableTightkID}, we make use of the following lemma, whose proof we defer to the end of the section.
\begin{lemma} \label{thm:kID}
    Consider the special case with $n$ agents and $k$ identical goods. The optimal omniscient approximation $\apxcs(n) \leq \tfrac{1}{(1-\frac{1}{\sqrt{2 \pi k}})} \cdot (1 - \ln \epsilon)(1+[\frac{n}{k}-1] \, \epsilon)$ for any $\epsilon > 0$. 
\end{lemma}

\Cref{thm:kID} is particularly useful when $k = \Theta(n)$. For a constant $\epsilon$, $\tfrac{1}{(1-\frac{1}{\sqrt{2 \pi k}})}$ is $(1 + o(1))$ while $(1 - \ln \epsilon)(1+[\frac{n}{k}-1] \, \epsilon)$ is a constant.

\begin{proof}[Proof of \Cref{thm:TractableTightkID}]
    We show the theorem separately for the three cases, when $k = \Theta(1)$, $k = \omega(1)$ but $o(n)$, and $k = \Theta(n)$.
    \begin{itemize}
        \item $k = \Theta(1)$: The proof is a direct consequence of \Cref{thm:Tight}, where we show $\apxcs(n) \leq (1 + o(1)) \ln n = (1 + o(1)) \ln \tfrac{n}{k}$. The equality follows since $k$ is a constant.
        \item $k = \omega(1)$, but $o(n)$: By setting $\epsilon = \frac{1}{\frac{n}{k} (1 + \ln \frac{n}{k})}$ in \Cref{thm:kID}, we get
        $$\apxcs(n) \leq \big(\tfrac{1}{1 - \frac{1}{\sqrt{2 \pi k}}}\big) \times (1 + \frac{1}{1 + \ln \frac{n}{k}} \, [1 - \frac{k}{n}]) \times (1 + \ln \tfrac{n}{k} + \ln (1 + \ln \tfrac{n}{k}))$$
        The first term is $(1 + o(1))$ since $k = \omega(1)$. The second and the third terms are $(1 + o(1))$ and $(1 + o(1)) \ln \frac{n}{k}$ respectively, since $k = o(n)$. Summarizing, $\apxcs(n) \leq (1 + o(1)) \ln \frac{n}{k}$.
        \end{itemize}
        The above two cases conclude the proof of the first part of the theorem. The next case clinches the proof of the second part.
        \begin{itemize}
        \item $k = \Theta(n)$: We set $\epsilon = \frac{1}{\frac{n}{k} (1 + \ln \frac{n}{k})}$ as in the previous case. Once again $\big(\tfrac{1}{1 - \frac{1}{\sqrt{2 \pi k}}}\big) = (1 + o(1))$ since $k = o(1)$. We plot the ratio between our upper bound and the lower bound $(0.577 + \ln \frac{n}{k})$ from \Cref{thm:LowerBound} as a function of $\tfrac{k}{n}$ in \Cref{fig:CompFigOPT} to observe that the ratio is at most $2.27$, exactly as needed. We skip a more detailed numerical analysis substantiating this observation. \qedhere
    \end{itemize}
\end{proof}

We modify the lower bound from \citet{QV23} sketched in \Cref{rem:Tight} to get an analogous result for the $k$-identical goods environment.

\begin{lemma} \label{thm:LowerBound}
    Consider the $n$-agent $k$-identical goods environment with values drawn IID from the exponential distribution with mean $1$. Assume, for convenience, that $\frac{n}{k}$ is an integer. The optimal omniscient approximation $\apxcs(n) \geq H_{\frac{n}{k}} \geq (0.577 + \ln \frac{n}{k})$.
\end{lemma}
\begin{proof}
    The exponential distribution with mean $1$ has a constant marginal consumer surplus curve $\util(q) = 1$, and thus, the optimal consumer surplus mechanism allocates to some set of $k$ agents to get an expected consumer surplus equal to $k$.

    Now, consider the following mechanism. Group agents into $k$ batches of $\frac{n}{k}$ agents each, and award the agent with the highest value in each group. Each group is now identical to a single item environment with $\frac{n}{k}$ agents, and thus, generates an expected surplus of $H_{\frac{n}{k}}$ (see \Cref{rem:Tight} on the example sketched in \citealp{QV23}). The $k$ batches together produce an expected surplus equal to $k \cdot H_{\frac{n}{k}}$. The optimal mechanism creates a larger surplus and hence, $\apxcs(n)$ is at least $H_{\frac{n}{k}} \geq (0.577 + \ln \frac{n}{k})$.
\end{proof}

Thus, the upper bound from \Cref{thm:TractableTightkID} is tight up to $(1 + o(1))$ terms when $k = o(n)$ and up to a multiplicative $2.27 \, (1 + o(1))$ when $k = \Theta(n)$. The best known bound for the optimal omniscient approximation is due to \citet{HR08}, where they show $\apxcs(n) \leq \tfrac{2}{\ln 2} \, (\ln 2 + \ln \frac{n}{k}) \approx 2.88 (\ln 2 + \ln \frac{n}{k})$. Apart from closing the $\tfrac{2}{\ln 2}$ gap between their upper bound and the lower bound when $k = o(n)$, our bound is an improvement by a factor of at least $1.36 (1 - o(1))$ when $k = \Theta(n)$ (\Cref{fig:CompFigHR}).

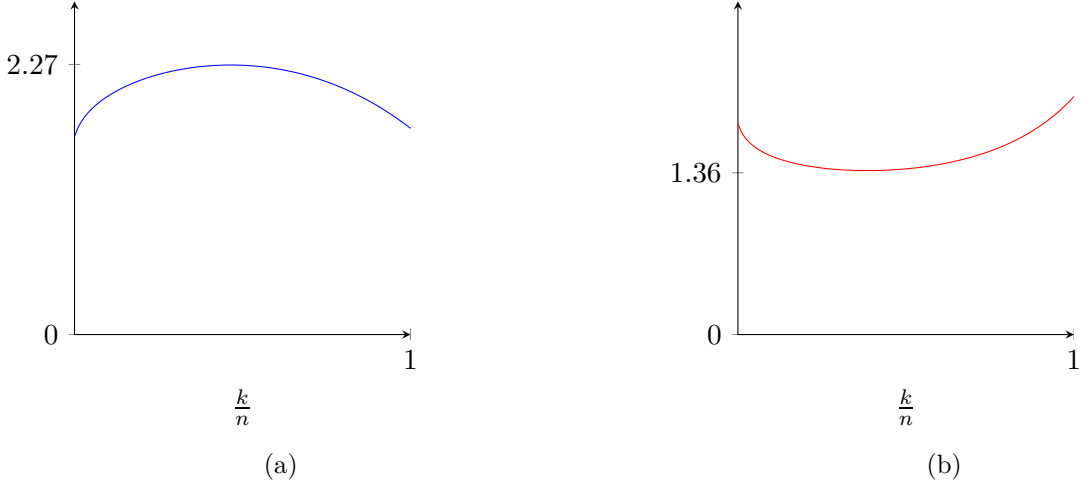
\begin{figure}
    \centering
    \begin{subfigure}[b]{0.45\textwidth}
\begin{tikzpicture}
\begin{axis}[
    xlabel=$\frac{k}{n}$,
    domain=0:1,
    xtick = {0, 1},
    ytick = {0, 2.27},
    ymin = 0,
    ymax = 2.8,
    samples=100,
    axis lines=left,
    clip=false,
    width=6cm,
    height=6cm,
    legend style={at={(0.95,0.95)}, anchor=north east}
]
\addplot[blue] {((1 + 1/(1 - ln(x))*(1-x))*(1-ln(x) + ln(1-ln(x)))) / (0.577 - ln(x))};
\end{axis}
\end{tikzpicture}
        \caption{}
        \label{fig:CompFigOPT}
    \end{subfigure}
    \hspace{1 cm}
    \begin{subfigure}[b]{0.45\textwidth}
\begin{tikzpicture}
\begin{axis}[
    xlabel=$\frac{k}{n}$,
    domain=0:1,
    xtick = {0, 1},
    ytick = {0, 1.36},
    ymin = 0,
    ymax = 2.8,
    samples=100,
    axis lines=left,
    clip=false,
    width=6cm,
    height=6cm,
    legend style={at={(0.95,0.95)}, anchor=north east}
]
\addplot[red] {((2/ln(2))*(ln(2)-ln(x))) / ((1 + 1/(1 - ln(x))*(1-x))*(1-ln(x) + ln(1-ln(x))))};
\end{axis}
\end{tikzpicture}
        \caption{}
        \label{fig:CompFigHR}
    \end{subfigure}
    \caption{We compare our upper bound against the lower bound and the best known upper bound when $k = \Theta(n)$. Figure (a) plots the ratio between our upper bound from \Cref{thm:TractableTightkID}, normalized by $(1 + o(1))$ and our lower bound from \Cref{thm:LowerBound}, $(0.577 + \ln \frac{n}{k})$. The maximum value of the curve is less than $2.27$, and hence, our upper bound can potentially be improved by a factor of at most $2.27 \, (1 + o(1))$. Figure (b) plots the ratio between the current best known bound $\tfrac{2}{\ln 2} \, (\ln 2 + \ln \frac{n}{k})$ and our bound from \Cref{thm:TractableTightkID}, again normalized by $(1 + o(1))$. The minimum value taken by the ratio is greater than $1.36$, and hence, our bound improves the best known bound by a factor of at least $1.36 (1 - o(1))$.}
    \label{fig:Comp}
\end{figure}

We conclude the section with the proof of \Cref{thm:kID}. We will adopt an approach similar to the proof of \Cref{thm:GeneralFeasibleMultiAgentEnvironments}. Increasing the allocation to the agents with a small quantile so that $y_i(q_i) = y_i(0)$ for $q_i \in [0, \epsilon]$ vastly increases the consumer surplus of $i$ without noticeably decreasing the surplus of the mechanism. We consider online posted-price mechanisms with the constraint $y_i'(q_i) = 0$ whenever $q_i \leq \epsilon$ for all agents $i$. Posted-price mechanisms are known to achieve a near optimal surplus and we bound the loss in surplus due to the additional constraint.

%We will find a mechanism such that, for every agent $i$, the interim allocation rule $y_i$ satisfies $y_i'(q_i) = 0$ for $q_i \in [0, \epsilon]$ and also obtains a near optimal surplus. We will propose an online price-posting mechanism which has the required property. In the online price-posting environment, the mechanism posts a price to each agent and determines whether to allocate to the agent before interacting with the next agent. For the purpose of this proof, we will let the mechanism decide the order in which it interacts with the agents (a.k.a, the sequential version).

Consistent with prior literature on online sequential price-posting \citep{CHMS10, Yan11} we will approximate the performance of the mechanism against the optimal ex-ante relaxed mechanism $\EAR$. In the ex-ante relaxation, the feasibility constraint is relaxed to bind ex-ante rather than ex-post. In other words, the expected number of agents served by the mechanism must be at most $k$. The optimal ex-ante mechanism obtains a larger surplus than $\OPTOM$, since $\OPTOM$ is also a feasible ex-ante mechanism. Note that, if agent $i$ is being allocated with probability $q_i$ by $\EAR$, it would allocate to the $q_i$ smallest quantiles (quantiles corresponding to the largest values) without violating the ex-ante constraint, thereby getting an expected surplus $\cumval_i(q_i)$ from agent $i$.

Let $\simplex$ be the set of all ex-ante feasible allocations $\vec{q} = (q_1, \dots, q_n)$ for selling up to $k$ items such that $q_1, \dots, q_n \in [0, 1]$ and $\sum_{i= 1}^n q_i \leq k$. Let $\EAR(\vec{q})$ be the ex-ante relaxed mechanism that allocates to the $q_i$ smallest quantiles of agent $i$. Note that $\EAR(\vec{q})$ is well-defined even if $\vec{q} \not \in \simplex$ as long as $q_1, \dots, q_n \in [0,1]$. We will abuse notation to denote the surplus from $\EAR(\vec{q})$ by $\EAR(\vec{q})$. Thus, $\EAR = \max_{\vec{q} \in \simplex} \EAR(\vec{q})$.

\citet{Yan11} showed that the following sequential posted-price mechanism $\GSP(\vec{q})$ has a surplus comparable to $\EAR(\vec{q})$ for all $\vec{q} \in \simplex$.
\begin{enumerate}
    \item Order agents in decreasing order of $\frac{\cumval_i(q_i)}{q_i}$, the expected surplus conditioned on selling to agent $i$.
    \item In the order above, offer agent $i$ a price $\val_i(q_i)$ until all $n$ agents interact with the mechanism or until all $k$ goods are sold.
\end{enumerate}
As usual, we will abuse notation to denote the surplus of $\GSP(\vec{q})$ by $\GSP(\vec{q})$.

\begin{theorem}[\citealp{Yan11}] \label{thm:Yan}
    For $\vec{q} \in \simplex$, $\tfrac{1}{(1 - \frac{1}{\sqrt{2 \pi k}})} \times \GSP(\vec{q}) \geq \EAR(\vec{q})$.
\end{theorem}

We propose a sequential-posted-price-mechanism with near optimal surplus and interim allocation rules satisfying the additional properties outlined above. For a vector $\vec{q} \in \simplex$, let the $\epsilon$-inflated sequential posted-price mechanism $\epsGSP{\vec{q}}$ denote $\GSP(\vec{p})$, where $p_i = \max \{q_i, \epsilon\}$. Observe that if $\epsGSP{\vec{q}}$ has an unsold good while interacting with agent $i$, it always sells the good if the agent has a quantile $q_i \leq \epsilon$. Thus, the interim allocation rule $y_i$ is a constant function in $[0, \epsilon]$ and $y_i'(q) = 0$ in this range. Next, we argue that $\epsGSP{\vec{q}}$ approximates the surplus of $\EAR(\vec{q})$.

\begin{lemma} \label{thm:epsSeqnearOPT}
    For all $\vec{q} \in \simplex$, $\EAR(\vec{q}) \leq \tfrac{(1+[\frac{n}{k}-1] \, \epsilon)}{(1-\frac{1}{\sqrt{2 \pi k}})} \epsGSP{\vec{q}}$
\end{lemma}
\begin{proof}
     The proof proceeds in two steps. First, for $\vec{p}$ with $p_i = \max \{q_i, \epsilon\}$, we show that $\vec{p}$ is not too far away from the set of ex-ante feasible allocations $\simplex$. We then show that the projection $\vec{t}$ of $\vec{p}$ back onto $\simplex$ satisfies $\EAR(\vec{q}) \leq \EAR(\vec{p}) \leq \tfrac{k+(n-k) \, \epsilon}{k} \times \EAR(\vec{t})$ and $\GSP(\vec{t}) \leq \GSP(\vec{p}) = \epsGSP{\vec{q}}$.

    \textbf{Step 1.} $\sum_{i = 1}^n p_i \leq k + (n-k) \, \epsilon$.

    Without loss of generality, we assume $\sum_{i = 1}^n q_i = k$. Increasing $q_i$ so that their sum equals $k$ would not decrease the value of $\sum_{i = 1}^n p_i$. The claim reduces to showing the maximum value of $\sum_{i=1}^n (p_i - q_1)$ is at most $(n - k) \, \epsilon$.
    
    Note that $p_i - q_i \leq \epsilon$, equality holding exactly when $q_i = 0$. Thus, whenever $q_i < \epsilon$, decreasing them to zero increases the value of $\sum_{i = 1}^n (p_i - q_i)$. However, this would contradict $\sum_{i = 1}^n q_i = k$. For $q_i \geq \epsilon$, $p_i - q_i = 0$. Increasing $q_i$ to $1$ whenever $q_i \geq \epsilon$ does not change the value of $\sum_{i =1}^n (p_i - q_i)$, but allows other values $q_j$ to be set to zero without violating $\sum_{i = 1}^n q_i = k$. A simple greedy argument would suggest that $\sum_{i = 1}^n (p_i - q_i)$ is maximized when $q_i = 1$ for $1 \leq i \leq k$ and $q_i = 0$ when $i > k$, in which case, $\sum_{i = 1}^n (p_i - q_i) < (n-k) \, \epsilon$.

    %We will argue that the sum $ k = \sum_{i = 1}^n q_i$ increases by at most $(n-k) \, \epsilon$ after updating $p_i = \max \{q_i, \epsilon\}$. The maximum increase occurs when $q_i = 0$ for $1 \leq i \leq n-k$ and $q_i = 1$ for $n-k < i \leq n$, where the sum increases exactly by $(n-k) \, \epsilon$. Note that only the agents with an allocation probability $q_i < \epsilon$ have their probability of allocation increased after updating $q_i$ to $p_i$. In order to maximize the increase in the sum, these agents should have an allocation probability equal to zero. On the other hand, agents with $q_i \geq \epsilon$ should have an allocation probability equal to $1$ in order to accommodate more number of agents with $q_j = 0$ while maintaining the sum of all allocation probabilities to be exactly $k$. A formal greedy argument can be made substantiating this intuition.

    \textbf{Step 2.} Order (and re-index) agents in decreasing order of $\cumval_i(p_i)$. Consider the ex-ante allocation $\tfrac{k}{k + (n-k) \, \epsilon} \cdot \EAR(\vec{p})$ that offers agent $i$ with $\tfrac{k}{k+(n-k) \, \epsilon}$ fraction of a good whenever $q_i \leq p_i$. This has a surplus $\tfrac{k}{k + (n-k) \, \epsilon} \cdot \EAR(\vec{p}) = \sum_{i=1}^n \tfrac{k}{k + (n-k) \, \epsilon} \cdot \cumval_i(p_i)$ and allocates up to $k$ units, since $\sum_{i = 1}^n p_i \leq k + (n-k) \, \epsilon$.  Let $\vec{t}$ be an alternate ex-ante feasible allocation that allocates $t_i = p_i$ for agents $1, \dots, m$ for some $m \leq n$ until all $k$ items are allocated. $\vec{t}$ redistributes goods from agents with a smaller surplus $\cumval_i(p_i)$ allocated by $\tfrac{k}{k + (n-k) \, \epsilon} \cdot \EAR(\vec{p})$ to agents with a larger $\cumval_i(p_i)$ to get a surplus $$\EAR(\vec{t}) = \sum_{i = 1}^m \cumval_i(p_i) \geq \sum_{i=1}^n \tfrac{k}{k + (n-k) \, \epsilon} \cdot \cumval_i(p_i) = \tfrac{k}{k + (n-k) \, \epsilon} \cdot \EAR(\vec{p})$$
    
    %Then,
    %$$\tfrac{k}{k + (n-k) \, \epsilon} \times \EAR(\vec{p}) \leq \EAR(\vec{t})$$
    
    %Let $n'$ be the smallest index such that $\sum_{i = 1}^{n'} p_i \geq k$. Set $t_i = p_i$ for $1 \leq i < n'$, $t_i = 0$ for $n'+1 \leq i \leq n$ and $t_{n'}$ such that $\sum_{i = 1}^{n'} t_i = k$. Denote the vector of truncated probabilities $(t_1, \dots, t_n)$ by $\vec{t}$. Then,
    %$$\EAR(\vec{p}) \leq \tfrac{k + (n-k) \, \epsilon}{k} \times \EAR(\vec{t})$$
    
    %$\tfrac{\EAR(\vec{p})}{k + (n-k) \, \epsilon} = \sum_{i = 1}^n \tfrac{p_i}{k + (n-k) \, \epsilon} \cdot \frac{\cumval(p_i)}{p_i}$ is the weighted sum of $\tfrac{\cumval(p_i)}{p_i}$ with weights $\tfrac{p_i}{k + (n-k) \, \epsilon}$, while $\tfrac{\EAR(\vec{t})}{k}$ is the weighted sum of $\tfrac{\cumval(t_i)}{t_i}$ with weights $\tfrac{t_i}{k}$. The total weights add up to at most $1$ in the former, while the sum of weights equals $1$ in the latter. Further, the weights are more concentrated on indices with a larger $\tfrac{\cumval(p_i)}{p_i}$ (i.e, smaller indices) in the latter. Hence, $\tfrac{\EAR(\vec{p})}{k + (n-k) \, \epsilon} \leq \tfrac{\EAR(\vec{t})}{k}$.

    \textbf{Step 3.}$\GSP(\vec{t}) \leq \GSP(\vec{p})$
    
    $\GSP{\vec{t}}$ allocates the first $m$ agents with probability $p_i$ and allocates the others with probability $0$. $\GSP(\vec{p})$ allocates the first $m$ agents with probability $p_i$ and further allocates any unsold good to the others. Thus, $\GSP(\vec{t}) \leq \GSP(\vec{p})$.
    
    %$\GSP(\vec{t})$ allocates the first $n'$ agents and allocates to the rest with probability zero. $\GSP(\vec{p})$ allocates to the first $n'$ agents with the same probabilities as $\GSP(\vec{t})$, and allocates with some more probability to the remaining agents. Thus, $\GSP(\vec{p})$ only generates more surplus than $\GSP(\vec{t})$.

    Summarizing,
    \begin{equation*}
        \notag
        \begin{split}
            \EAR(\vec{q}) &\leq \EAR(\vec{p}) \\
            &\leq \tfrac{k+(n-k) \, \epsilon}{k} \times \EAR(\vec{t}) =  (1 + [\frac{n}{k} - 1] \, \epsilon) \times \EAR(\vec{t})\\
            &\leq \frac{(1 + [\frac{n}{k} - 1] \, \epsilon)}{(1 - \frac{1}{\sqrt{2 \pi k}})} \times \GSP(\vec{t}) \\
            &\leq \frac{(1 + [\frac{n}{k} - 1] \, \epsilon)}{(1 - \frac{1}{\sqrt{2 \pi k}})} \times \GSP(\vec{p}) = \frac{(1 + [\frac{n}{k} - 1] \, \epsilon)}{(1 - \frac{1}{\sqrt{2 \pi k}})} \times \epsGSP{\vec{q}}
        \end{split}
    \end{equation*}
    The third inequality follows from \Cref{thm:Yan}
\end{proof}

\begin{proof}[Proof of \Cref{thm:kID}]
    Let $\vec{q} \in \simplex$ be such that $\EAR = \EAR(\vec{q})$. From the review on sequential-price-posting and \Cref{thm:epsSeqnearOPT},
    $$\OPTOM \leq \EAR = \EAR(\vec{q}) \leq \tfrac{(1+[\frac{n}{k}-1] \, \epsilon)}{(1-\frac{1}{\sqrt{2 \pi k}})} \times \epsGSP{\vec{q}}$$
    Further, $\epsGSP{\vec{q}}$ has a constant interim allocation rule for agent $i$ when its quantile $q_i$ is in the range $[0, \epsilon]$. Arguments identical to the proof of \Cref{thm:GeneralFeasibleMultiAgentEnvironments} can be used to conclude
    $$\apxcs(n) = \frac{\OPTOM}{\OPTCS} \leq \tfrac{(1+[\frac{n}{k}-1] \, \epsilon)}{(1-\frac{1}{\sqrt{2 \pi k}})} \cdot \frac{\epsGSP{\vec{q}}}{\OPTCS} \leq  \tfrac{1}{(1-\frac{1}{\sqrt{2 \pi k}})} \cdot (1 - \ln \epsilon)(1+[\frac{n}{k}-1] \, \epsilon)\qedhere$$
\end{proof}

\section{Pen Testing Corollaries from Deferred-Acceptance Mechanisms} \label{sec:DAAPrev}
There are many environments in which deferred-acceptance mechanisms are known to be good.  By \Cref{thm:IntroSurplusDA}, these imply good pen testing algorithms (up to an additional $\apxcs(n)$ factor in the omniscient approximation, i.e, the ratio of the performances of the omniscient algorithm that knows the amount of ink in each pen and our proposed pen testing algorithm). In this section, we discuss a few notable examples.    A summary of these results was given previously in \Cref{table:Summary}.

Deferred-acceptance mechanisms that achieve the ex-post surplus optimal outcome are known for various feasibility constraints. While the simple auction that uniformly increases the price until exactly $k$ bidders remain active is surplus optimal in the $k$-identical goods environment, \citet{MS14} and \citet{BdVSV11} generalize the mechanism for matroid feasibility constraints. We get the following corollaries from the above auctions.
\begin{corollary} \label{thm:MatroidVirtualSurplus}
    For a pen testing environment with a matroid feasibility constraint, there exists a pen testing algorithm with an omniscient approximation ratio $(1+o(1)) \ln n$.
\end{corollary}

\begin{corollary} \label{thm:kPens}
    When choosing any $k$ of the $n$ pens is feasible, there exists a pen testing algorithm with an omniscient approximation ratio $(1 + o(1)) \ln \frac{n}{k}$ when $k = o(n)$ and $2.27 \, (1 + o(1)) \, (0.577 + \ln \frac{n}{k})$ when $k = \Theta(n)$.
\end{corollary}

\citet{DGT14} initiated the study of prior-free deferred-acceptance approximation mechanisms in environments where deferred-acceptance mechanisms are known to be suboptimal. Their bounds can be improved in settings like ours where prior distributions of the agents' values are known. For general downward-closed constraints with a prior distribution on values, \citet{FGGS22} give a $O(\log \log m)$ approximation to the optimal surplus, where $m$ is the number of maximal feasible sets.  Note that $m \leq 2^n$ (every subset of agents might be feasible) and hence, $\log \log m$ is at most $\log n$. Thus, we have a poly-logarithmic approximate pen testing algorithm for any downward-closed feasibility environment.

\begin{comment}
\begin{corollary} \label{thm:GeneralFeasibleAscending}
    For a pen testing environment with a downward-closed feasibility constraint, there exists an algorithm with an omniscient approximation ratio $O(\log^2 n)$.
\end{corollary}    
\end{comment}

Note, however, that the pen testing model is inherently downward-closed.  Specifically, suppose an algorithm wanted to select a set $P \subseteq \overline{P}$, but $P$ is not feasible while $\overline{P}$ is. Since all pens have non-negative residual ink, even the ones that failed their tests, there is no loss in selecting $\overline{P}$ instead of the set $P$. Thus, near optimal pen testing algorithms for downward-closed constraints can be extended to give near optimal algorithms for general combinatorial constraints, giving the same performance guarantee as the downward-closed environment.

\begin{comment}
The problem of designing pen testing algorithms for general combinatorial environments can be reduced to designing algorithms for downward-closed constraints. For a feasibility constraint $\feasibility$, construct its downward-closure $\dcfeasibility$ by adding all subsets of $p$ to the constraint for all $p \in \feasibility$. We have a $O(\log^2 n)$-approximate algorithm for the constraint $\dcfeasibility$. Let $p$ be the output of the algorithm. If $p$ is feasible ($p \in \feasibility$), we already have a $O(\log^2 n)$-approximate selection. Otherwise, $p$ was added to the downward closure of $\feasibility$ because of the existence of $\overline{p} \in \feasibility$ such that $p \subseteq \overline{p}$. Choose the set $\overline{p}$ instead. Since the ink levels are non-negative, the remaining ink in $\overline{p}$ is at least the remaining ink in $p$. Since the latter is a $O(\log^2 n)$-approximation, so is the former. 
\end{comment}
\begin{corollary} \label{thm:CombinatorialPenTesting}
    For any combinatorial pen testing environment, there exists a pen testing algorithm with an omniscient approximation ratio $O(\log n \log \log m) = O(\log^2 n)$.
\end{corollary}
\begin{remark} \label{rem:GeneralFeasibilityAscendingAuctions}
    Note that this reduction cannot be used to design deferred-acceptance mechanisms for general feasibility constraints. Allocating to agents that have dropped out of the auction can make the mechanism non-truthful (i.e, staying active till the price reaches the agent's value and dropping out subsequently might not be in the agent's best interests).
\end{remark}

The bound for downward-closed and general combinatorial environments can be improved upon in special cases; for example, in \Cref{sec:Knapsack} we give a natural deferred-acceptance mechanism for knapsack constraints that is a 2-approximation. For this knapsack problem, the agents have sizes along with values and feasible subsets are precisely those with a total size at most the knapsack capacity.

\begin{corollary} \label{thm:Knapsack}
    For a pen testing environment with a knapsack feasibility constraint, there exists a pen testing algorithm with an omniscient approximation ratio $2\,(1+o(1)) \ln n$.
\end{corollary}

Online pricing mechanisms are a special case of deferred-acceptance mechanisms; thus, \Cref{thm:IntroSurplusDA} converts online pricing mechanisms with good surplus into good online pen testing algorithms.  For online pen testing problems, every pen can be tested exactly once and must either be selected or discarded immediately after testing and before testing another pen. The algorithm might have the ability to determine the order of testing pens (the sequential version) or the order might be adversarially determined (the oblivious version).  

For the sequential pricing problem, \citet{CHMS10} gave an $\frac{e}{e-1}$-approximation mechanism for the single-item environment.  \citet{Yan11} connected the sequential pricing problem to the correlation gap and generalized the $\frac{e}{e-1}$-approximation to matroid environments.  Both of these papers considered optimizing revenue, but the results can be easily adapted to optimize surplus.

\begin{corollary} \label{thm:PenCorrelationGap}
    In the online sequential pen testing problem with a matroid feasibility constraint, there exists a pen testing algorithm that achieves an omniscient approximation ratio $\frac{e}{e-1} \cdot (1+o(1)) \ln n$.
\end{corollary}

\citet{HKS07} and \cite{CHMS10} adapt the prophet inequality of \citet{ESC84} to give a non-adaptive $2$-approximation for the oblivious price posting problem with $1$ good.  \citet{KW12} extend this 2-approximation to matroid environments and any arrival order of agents but with adaptive prices.

\begin{corollary} \label{thm:PenProphet}
    In the online oblivious pen testing problem to select a single pen, there exists a non-adaptive pen testing algorithm that achieves an omniscient approximation ratio ${2(1 + o(1)) \ln n}$.    
\end{corollary}
\begin{corollary} \label{thm:MatroidPenProphet}
    In the online oblivious pen testing problem with a matroid feasibility constraint, there exists a pen testing algorithm that achieves an omniscient approximation ratio $2(1 + o(1)) \ln n$.    
\end{corollary}

The bounds of \Cref{thm:PenCorrelationGap} and \Cref{thm:PenProphet} improve on the online pen testing bounds of \citet{QV23}.

\section{The Online IID Environment} 
\label{sec:IID}

Consider the special case of the online pen testing problem where the ink levels are drawn independently from the same distribution. We move away from the reduction framework of \Cref{thm:IntroSurplusDA} to obtain a better bound for this environment.

\begin{theorem} \label{thm:IID}
In online single-item environments with IID agents, there exists a price-posting strategy with an omniscient approximation at most $(1 + o(1)) \ln n$. Equivalently, in the online pen testing problem with IID ink levels, there exists an algorithm that achieves $\pi(n) \leq {(1 + o(1)) \ln n}$.
\end{theorem}
First, note that the above theorem is tight. In \Cref{rem:Tight}, we saw an instance (originally from \citealp{QV23}) with $n$ IID agents whose values are drawn from the exponential distribution with consumer surplus only an $H_n$ approximation to the optimal surplus. The optimal online consumer surplus cannot do better. Hence, the $(1+o(1)) \ln n$ approximation ratio in \Cref{thm:IID} is tight.

Before proving \Cref{thm:IID}, we briefly discuss the guarantee obtained through the reduction framework of \Cref{thm:IntroSurplusDA}. For the online environment with IID agents, it is known that $\apxda(n) = \frac{e}{e-1}$ \citep{CHMS10}. Thus, the framework would give us an $\frac{e}{e-1} \cdot (1+o(1)) \ln n$ bound. The direct analysis below gives a tight $H_n + 1 = (1 + o(1)) \ln n$ bound up to an additive factor $1$.

\begin{proof}[Proof of \Cref{thm:IID}]
    Similar to the proof of \Cref{thm:Main}, we will phrase our program as an optimization problem and compute the optimal solution.
    
    The surplus optimal auction always allocates the good to the agent with the highest value (smallest quantile). Let $\val_{max}$ be the random variable denoting the value of the winner.
    $$Pr(\val_{\max} \leq \val(t)) = (1-t)^n$$
    Thus, the expected surplus of the surplus optimal mechanism equals
    \begin{align}
        \int_0^1 \val(t) \times n(1-t)^{n-1} dt &= \int_0^1 \int_t^1 \frac{\util(r)}{r} dr \times n(1-t)^{n-1} dt \nonumber \\
        &= \int_0^1 \frac{\util(t)}{t} dt - \int_0^1 \frac{\util(t)}{t}(1-t)^n dt \nonumber \\
        &= \int_0^1 \frac{1-(1-t)^n}{t} \util(t) dt \nonumber \\
        &= \cumutil(1) + \int_0^1 \frac{1-(1-t)^n - nt(1-t)^{n-1}}{t^2} \cumutil(t) dt \nonumber
    \end{align}
    The first equality follows by integrating from $\util(t) = -t\val'(t)$ and $\val(1) = 0$. We integrate by parts for the second and fourth equation.
    
    Next, we look at the consumer surplus through price-posting. For a price $\val(q)$, we allocate the good if at least one agent has a quantile in $[0, q]$. Conditioned on selling the good, we get an expected consumer surplus of $\frac{\cumutil(q)}{q}$. Thus, the expected consumer surplus equals
    $$\frac{1-(1-q)^n}{q} \times \cumutil(q)$$
    As before, we first show \Cref{thm:IID} for consumer surplus regular distributions. Without loss of generality, assume that the maximum achievable consumer surplus through anonymous price posting is at most $1$. We want to find the consumer surplus regular distribution with the largest surplus.
    $$\max \cumutil(1) + \int_0^1 \frac{1-(1-t)^n - nt(1-t)^{n-1}}{t^2} \cumutil(t) dt$$
    subject to
    $$\frac{1-(1-q)^n}{q} \times \cumutil(q) \leq 1 \text{ for all } q \in [0, 1]$$
    $$\cumutil \text{ is concave}$$
    First, observe that $C(q) = \frac{q}{1-(1-q)^n}$ is convex (see \Cref{thm:CisConvex} in \Cref{sec:Omitted}). We want to fit a concave curve $\cumutil$ under the convex curve $C$. Let $\overline{C}_{q}$ be the tangent to $C$ at $q$. $\overline{C}_{\hat{q}}$ is a feasible solution for $\cumutil$ for all $\hat{q} \in [0, 1]$. Further, given that $\cumutil$ touches the curve $C$ at $\hat{q}$, $\cumutil(t) \leq \overline{C}_{\hat{q}}(t)$ for all $t \in [0, 1]$ (from the concavity of $\cumutil$). Also, it is clear that the objective increases with a pointwise increase in $\cumutil$. Thus, we conclude that the optimal solution to the program is achieved at $\cumutil = \overline{C}_{\hat{q}}$ at some $\hat{q} \in [0, 1]$. In particular, the worst-case ratio occurs when $\cumutil$ is a straight-line.
    
    Without loss of generality, we can normalize the slope to $1$. Let $\cumutil(q) = q + a$. The surplus equals
    \begin{equation}
        \notag
        \begin{split}
            1 + a + \int_0^1 \frac{1-(1-t)^n - nt(1-t)^{n-1}}{t} dt + a \times \int_0^1 \frac{1-(1-t)^n - nt(1-t)^{n-1}}{t^2} dt &= H_n + an
        \end{split}
    \end{equation}
    See \Cref{thm:IntegralEquality} in \Cref{sec:Omitted} for the proof of the above equality.
    
    By posting a price $\val(q)$, we generate a consumer surplus
    $$\frac{1-(1-q)^n}{q} \times \cumutil(q) = (1 + (1-q) + \dots + (1-q)^{n-1}) \times (q + a)$$
    If $a \geq \frac{1}{n}$, the ratio between the surplus and the consumer surplus is at most $H_n + 1$ by setting $q = 0$. If $a < \frac{1}{n}$, the surplus is less than $H_n + 1$, and by setting $q = 1$, the consumer surplus is at least $1$. Thus, the worst-case ratio between surplus and consumer surplus is at most $H_n + 1 = (1 + o(1)) \ln n$.
    
    We now show the result for all distributions. We rewrite the optimal surplus as follows.
    $$\int_0^1 \val(t) \times n(1-t)^{n-1} = \Bigr[\cumval(t) \times n (1-t)^{n-1}\Bigl]_{t=0}^{t=1} + \int_0^1 \cumval(t) \times n(n-1)(1-t)^{n-2}dt$$
    The equality follows through integration by parts. From the above expression, it can be concluded that the optimal surplus increases with a pointwise increase in $\cumval$.
    
    Let $\icumval$ be the price-posting surplus curve of the distribution with a price-posting consumer surplus curve $\icumutil$. $\icumutil$ is pointwise larger than $\cumutil$, and hence, from \Cref{thm:IncRSCurve}, $\icumval$ is pointwise larger than $\cumval$. Hence, the optimal surplus is larger in the distribution with price-posting consumer surplus curve $\icumutil$. However, from \Cref{thm:EquivIroning}, the consumer surplus from posting prices is identical to both the distributions. Thus, the consumer surplus irregular distribution has a better ratio between surplus and consumer surplus than the consumer surplus regular distribution. Thus, the ratio between optimal surplus to optimal consumer surplus is at most $(1 + o(1)) \ln n$ for all distributions.
\end{proof}

\section{Conclusion}
In this section, we discuss the various advantages and disadvantages of using the reduction framework of \Cref{thm:IntroSurplusDA}. Connecting the pen testing problem to auction theory and building an easy-to-use reduction framework enable us to design near optimal pen testing algorithms for any general feasibility constraint. We are able to construct algorithms that are only a constant factor away from the lower bound described in \citet{QV23} for various feasibility constraints like matroids ($\apxda(n) = 1$) and knapsack ($\apxda(n) = 2$). However, there is a scope for improvement on the following two fronts:
\begin{enumerate}
    \item obtaining tight approximation guarantees, and
    \item obtaining approximation guarantees under other models of access to the prior distributions.
\end{enumerate}

\textbf{Tight Approximation Guarantees:} The approximation ratio of any pen testing algorithm is at least the optimal omniscient approximation $\apxcs(n)$. This follows since the optimal consumer surplus in the underlying auction environment is an upper bound on the performance of the optimal pen testing algorithm. The framework matches this bound up to a multiplicative factor of the approximation $\apxda(n)$ of deferred-acceptance mechanisms. Reducing this gap for various feasibility environments stands as an open problem.

For instance, consider the online environments discussed in the paper, where $\apxcs(n) = (1+o(1)) \ln n$. Our framework yielded an omniscient approximation $2(1+o(1)) \ln n$ in the oblivious version ($\apxda(n) = 2$) and $\tfrac{e}{e-1} \cdot (1+o(1)) \ln n$ in the sequential version ($\apxda(n) = \frac{e}{e-1}$). We conjecture that the approximation guarantees obtained by applying the reduction framework are not tight, and can be strengthened to $(1+o(1)) \ln n$ in both these environments, like in the IID setting (\Cref{sec:IID}).

\textbf{Other Models of Access to Distributions:} Our framework crucially makes use of virtual values, which in turn need complete knowledge of the distribution of ink levels. Thus, tailoring this approach to models where the algorithm gets partial access to the priors is challenging.

\citet{QV23} study two such models. They study the online oblivious problem where the algorithm gets access to just one sample from the distribution of each pen. They give a $O(\log n)$-approximation to the omniscient benchmark in this setting. They also look at the online secretary setting, where the pens arrive according to a permutation chosen uniformly at random. The quantity of ink in the pens are adversarially determined. The algorithm is told the quantity of ink in the pen with the largest ink level. They give a $O(\log n)$-approximation in this environment. Further, when the algorithm is given access to the ink levels in all the $n$ pens (with random arrival order), they tighten their bound to $O(\frac{\log n}{\log \log n})$. It is an open question to identify reduction frameworks for these different models of access to the distributions.

\bibliography{references}

\appendix

\section{A Discussion on the Surplus of the Single-Agent Optimal Consumer Surplus Auction} \label{sec:SurplusofOPTCS}
In \Cref{sec:SingleAgent}, we compare the optimal consumer surplus with an ex-ante allocation constraint $q$ to the optimal surplus. We now take a look at the surplus generated by the consumer surplus optimal auction.
\begin{corollary}
For an ex-ante allocation probability $q$, the ratio of the optimal surplus to the surplus from the consumer surplus maximizing auction is at most $1 - \ln q$.
\end{corollary}
The corollary follows since the surplus of the consumer surplus maximizing auction is only larger than the consumer surplus.

The above corollary is tight. This follows immediately from the observations below:
\begin{enumerate}
    \item For a distribution with a convex price-posting consumer surplus curve, the consumer surplus optimal mechanism is obtained by giving away the good for free with probability $q$. For such distributions, the consumer surplus of the optimal mechanism equals the surplus of the mechanism.
    \item The price-posting consumer surplus curve of the exponential distribution is a straight line. A distribution obtained by adding a small convex distortion to this curve has the same ironed price-posting consumer surplus curve as the exponential distribution and close to the same price-posting surplus curve as the exponential distribution. 
\end{enumerate}
The ratio of surplus in the surplus optimal mechanism to the consumer surplus optimal mechanism is close to $1 - \ln q$ for distributions described in bullet 2.

\section{Deferred-Acceptance Mechanism for Knapsack Constraints} \label{sec:Knapsack}
Consider an environment with a knapsack feasibility constraint. Each agent $i$ has a value $\val_i$ drawn independently from distribution $F_i$ and a size $s_i$. The goal is to chose a set $S$ of agents maximizing surplus such the total size $\sum_{i \in S} s_i$ of chosen agents is at most $C$.

We rely on the following folklore $2$-approximation algorithm for the knapsack problem.
\begin{algorithm} \label{alg:DAKnapsack}
Run the alternative with a better expected surplus:
\begin{enumerate}
    \item Bang-per-buck \citep{Dant57}: Sort agents in decreasing order of $\frac{\val_i}{s_i}$. Pick agents in this order until the knapsack is full.
    \item Max: Pick the agent with the largest value.
\end{enumerate}
\end{algorithm}
\cite{MS14} give deferred-acceptance implementations for both, bang-per-buck and max. Thus, the above is a deferred-acceptance mechanism.

\begin{lemma}
    \Cref{alg:DAKnapsack} is a $2$-approximation to the optimal surplus.
\end{lemma}
\begin{proof}
    Let OPT\textsuperscript{V} be the random variable denoting the optimal surplus. Fix some realization of values $\val_1, \val_2, \dots, \val_n$. Without loss of generality, assume $\frac{\val_1}{s_1} \geq \frac{\val_2}{s_2} \geq \dots \geq \frac{\val_n}{s_n}$. Let bang-per-buck pick agents $1$ through $t$ for this realization of values.

    Clearly, $\sum_{i = 1}^{t+1} \val_i \geq \OPTOM$. This is because the optimal algorithm that can pack fractional agents in the knapsack would greedily pick the first $t$ agents and pick some fraction of agent $t+1$. Thus, if $\val_{\max} = \max_{1 \leq i \leq n} \val_i$, $\sum_{i = 1}^t \val_i + \val_{\max} \geq \text{OPT}\textsuperscript{V}$. The surplus generated jointly by the two alternatives pointwise dominates the optimal surplus and hence dominates the optimal surplus in expectation.
    $$E_{\val_i \sim F_i}[\text{surplus(bang-per-buck)}] + E_{\val_i \sim F_i}[\text{surplus(max)}] \geq E_{\val_i \sim F_i}[\text{OPT}\textsuperscript{V}]$$
    We get the larger of the two terms in the left hand side, which must be at least $\tfrac{1}{2} E_{\val_i \sim F_i}[\text{OPT}\textsuperscript{V}]$.
\end{proof}

\section{Lemmas for Theorem 10} \label{sec:Omitted}
\begin{lemma} \label{thm:CisConvex}
    $\frac{q}{1-(1-q)^n}$ is convex in $q$.
\end{lemma}
\begin{proof}
    $$\frac{q}{1-(1-q)^n} = \frac{1}{\sum_{j = 0}^{n-1} (1-q)^j}$$
    $$\frac{d}{dq} \frac{1}{\sum_{j = 0}^{n-1} (1-q)^j} = \frac{\sum_{j = 0}^{n-2} (j+1)(1-q)^{j}}{\sum_{j = 0}^{n-1} (1-q)^j} \cdot \frac{1}{\sum_{j = 0}^{n-1} (1-q)^j}$$
    The second term in the product is increasing in $q$. We will show the first term is also increasing in $q$. The first term can be rewritten as
        $$\frac{\sum_{j = 0}^{n-2} (j+1)(1-q)^{j}}{\sum_{j = 0}^{n-1} (1-q)^j} = \frac{\sum_{j = 0}^{n-1} [(1-q)^j - (1-q)^{n-1}]}{\sum_{j=0}^{n-1} (1-q)^j} = 1 - n\frac{(1-q)^{n-1}}{1 + (1-q) + \dots + (1-q)^{n-1}}$$
    Further rewriting this, it becomes apparent that this is indeed increasing in $q$.
    $$\frac{\sum_{j = 0}^{n-2} (j+1)(1-q)^{j}}{\sum_{j = 0}^{n-1} (1-q)^j} = 1 - n \frac{1}{\frac{1}{(1-q)^{n-1}} + \dots + \frac{1}{(1-q)^0}}. \qedhere$$
\end{proof}

\begin{lemma} \label{thm:IntegralEquality}
    $$1 + a + \int_0^1 \frac{1-(1-t)^n - nt(1-t)^{n-1}}{t} dt + a \times \int_0^1 \frac{1-(1-t)^n - nt(1-t)^{n-1}}{t^2} dt = H_n + an$$
\end{lemma}
\begin{proof}
    The key is to rewrite $\frac{1-(1-t)^n - nt(1-t)^{n-1}}{t}$.
    \begin{equation}
    \notag
        \begin{split}
              \frac{1-(1-t)^n - nt(1-t)^{n-1}}{t} &= \sum_{j = 0}^{n-1} (1-t)^j - n(1-t)^{n-1} \\
              &= \sum_{j = 0}^{n-1} [(1-t)^j - (1-t)^{n-1}] \\
              &= t \times \sum_{j = 0}^{n-2} (j+1) \times (1-t)^j
        \end{split}
    \end{equation}
    Thus, the expression reduces to
    $$1 + a + \sum_{j = 0}^{n-2}\int_0^1  (j+1) \times t(1-t)^j dt + a \times \sum_{j = 0}^{n-2} \int_0^1 (j+1) \times (1-t)^j dt = H_n + an. \qedhere$$
\end{proof}
\end{document}